\documentclass[sigplan]{acmart}

\acmConference[]{In Submission}{2020}{}

\setcopyright{none}

\usepackage{booktabs}   
\usepackage{subcaption} 

\usepackage{xspace}
\usepackage{multirow}
\usepackage{algorithmic}
\usepackage[linesnumbered,ruled,vlined]{algorithm2e}
\usepackage{pgffor} 
\usepackage{mathtools}
\usepackage{amsmath}
\usepackage{amsthm}
\usepackage{enumitem}
\usepackage{caption}
\usepackage{pgfplots}
\pgfplotsset{
    compat=1.4,
    small,
    legend style={
        at={(0.99,0.99)},
        anchor=north east,
        font=\bfseries,
    },
    label style={font=\sffamily\small\bfseries}
}%

\usepackage{mdframed}
\usepackage{xcolor}
\usepackage{adjustbox}
\usepackage{listings}
   \lstdefinestyle{mystyle}{
      frame=single,
      framexleftmargin=0pt,
      commentstyle=\color{green},
      keywordstyle=\color{blue}\bfseries,
      numberstyle=\tiny\color{gray},
      stringstyle=\color{purple},
      basicstyle=\scriptsize\ttfamily\bfseries,
      breakatwhitespace=false,         
      breaklines=false,                 
      captionpos=b,                    
      keepspaces=true,     
      numbers=none,                    
      numbersep=4pt,                  
      showspaces=false,                
      showstringspaces=false,
      showtabs=false,                  
      tabsize=2,
      language=Java,
      escapechar=\%
   }
\usepackage{pifont}
\usepackage{tikz}
\usetikzlibrary{matrix}
\usepackage{bm}
\usepackage{flushend}

\definecolor{light-gray}{gray}{0.80}

\newcommand{\tool}{\textsc{NeurSA}\xspace}

\newcommand{\tabincell}[2]{\begin{tabular}{@{}#1@{}}#2\end{tabular}}

\newcommand{\wyc}[1]{{\color{red} \sf #1}}
\newcommand{\ignore}[1]{\iffalse #1 \fi}

\usepackage{rotating}

\newcommand{\etal}{\hbox{\emph{et al.}}\xspace}
\newcommand{\eg}{\hbox{\emph{e.g.}}\xspace}
\newcommand{\ie}{\hbox{\emph{i.e.}}\xspace}
 
\newcommand{\st}{\hbox{\emph{s.t.}}\xspace}
\newcommand{\wrt}{\hbox{\emph{w.r.t.}}\xspace}
\newcommand{\etc}{\hbox{\emph{etc.}}\xspace}

\definecolor{OliveGreen}{rgb}{0,0.6,0}

\DeclareMathOperator*{\maxx}{max}

\newcommand{\rBugs}{50\xspace}
\newcommand{\cBugs}{3\xspace}
\newcommand{\fBugs}{7\xspace}

\begin{document}

\title[Learning a Static Bug Finder from Data]{Learning a Static Bug Finder from Data}


\author{Yu Wang}
\orcid{nnnn-nnnn-nnnn-nnnn}             
\affiliation{
  \position{}
  \department{State Key Laboratory of Novel Software Technology}              
  \institution{Nanjing University}            
  \streetaddress{}
  \city{Nanjing}
  \state{Jiangsu}
  \postcode{210023}
  \country{China}                    
}
\email{yuwang_cs@nju.edu.cn}          

\author{Fengjuan Gao}
\orcid{nnnn-nnnn-nnnn-nnnn}             
\affiliation{
  \position{}
  \department{State Key Laboratory of Novel Software Technology}              
  \institution{Nanjing University}            
  \streetaddress{}
  \city{Nanjing}
  \state{Jiangsu}
  \postcode{210023}
  \country{China}                    
}
\email{fjgao@smail.nju.edu.cn}          

\author{Linzhang Wang}
\orcid{nnnn-nnnn-nnnn-nnnn}             
\affiliation{
  \position{}
  \department{State Key Laboratory of Novel Software Technology}              
  \institution{Nanjing University}            
  \streetaddress{}
  \city{Nanjing}
  \state{Jiangsu}
  \postcode{210023}
  \country{China}                    
}
\email{lzwang@nju.edu.cn}          

\author{Ke Wang}
\authornote{Ke Wang is the corresponding author.}          
\orcid{nnnn-nnnn-nnnn-nnnn}             
\affiliation{
  \position{}
  \department{Visa Research}              
  \institution{Visa Inc.}            
  \streetaddress{}
  \city{Palo Alto}
  \state{CA}
  \postcode{}
  \country{USA}                    
}
\email{kewang@visa.com}          



\begin{abstract}\label{sec:Abstract}


We present an alternative approach to creating static bug finders. 
Instead of relying on human expertise, we utilize deep neural networks
to train static analyzers directly from data. In particular, we frame the problem of bug finding as a classification task and train a classifier to differentiate the buggy from non-buggy programs using Graph Neural Network (GNN). Crucially, we propose a novel interval-based propagation mechanism that leads to a significantly more efficient, accurate and scalable generalization of GNN. 

We have realized our approach into a framework, \tool, and extensively evaluated it. In a cross-project prediction task, three neural bug detectors we instantiate from \tool are effective in catching null pointer dereference, array index out of bound and class cast bugs in unseen code. 
We compare \tool against several static analyzers (\eg Facebook Infer and Pinpoint) on a set of null pointer dereference bugs. Results show that \tool is more precise in catching the real bugs and suppressing the spurious warnings. We also apply \tool to several popular Java projects on GitHub and discover 50 new bugs, among which 9 have been fixed, and 3 have been confirmed.

\end{abstract}

\ccsdesc[500]{Software and its engineering~Software notations and tools}
\ccsdesc[500]{Software and its engineering~General programming languages}


\maketitle


\section{Introduction}\label{sec:Intro}

Static analysis is an effective technique to catch bugs early when they are cheap to fix.
Unlike dynamic analysis, static analysis reasons about every path in a program, offering formal guarantees for its run-time behavior. As an evidence of their increasing maturity and popularity, many static analyzers have been adopted by major tech companies to prevent bugs leaked to their production code. Examples include Google's Tricorder~\citep{Sadowski2015}, Facebook's Getafix~\citep{scott2019getafix} and Zoncolan, and Microsoft's Visual Studio IntelliCode.

Despite the significant progress, static analyzers suffer from several well-known issues. One, in particular, is the high false positive rate which tends to overshadow true positives and hurt usability. The reason for this phenomenon is well-known: all nontrivial program properties are mathematically
undecidable, meaning that automated reasoning of software generally must involve approximation. 
On the other hand, problems of false negatives also need to be dealt with. Recently,~\citet{Habib2018} investigated how effective the state-of-the-art static analyzers are in handling a set of real-world bugs. Habib~\etal show more than 90\% of the bugs are missed, exposing the severity of false negatives.

To tackle the aforementioned weaknesses, this paper explores an alternative approach to create static bug checkers---neural bug detection. Our observation is bug patterns exist even for those targeted by static analyzers. Therefore, machine learning techniques offer a viable solution. The challenge though is how to design a model that is effective in catching such bug patterns, which are non-trivial and can even be quite complex. Although Graph Neural Networks (GNN) have seen success across many program analysis tasks~\cite{li2015gated,allamanis2017learning,si2018learning}, they have been utilized only to catch specialized, relatively syntactic bugs in variable naming~\cite{allamanis2017learning}, hence limited impact. To elevate GNN's capability to a sufficient extent for detecting semantic bugs, we invent a novel, efficient propagation mechanism based on graph intervals. Our insight is specializing GNN to exploit the program-specific graph characteristics enhances their capacity of learning complex programming patterns. Specifically, our propagation mechanism forces GNN to attend to certain code constructs to extract deep, semantic program features. Moreover, our propagation mechanism operates on hierarchical graphs to alleviate the scalability issues GNN often encounters when learning from larger graphs.

We realize our approach into a framework, called \tool, that is general (\ie language-agnostic), and extensible (\ie not restricted to certain bug types). \tool performs inter-procedural analysis for bug finding. The framework first automatically scrapes training data across multiple project corpus, creating a large set of buggy and clean code examples; then trains a model to differentiate between these two; and finally uses the trained model for detecting bugs in previously unseen code. We present three neural bug detectors based on \tool targeting null pointer dereference, array index out of bound and class cast exceptions. Figure~\ref{fig:1}-\ref{fig:3} depict three real bugs (from the dataset introduced by~\citet{ye2014learning} and~\citet{just2014defects4j}) that are caught by our neural bug detectors. 
In principle, extending \tool to create new bug detectors only requires the bug type information as the entire workflow will be automatically performed by \tool.
In theory, compared to the soundy static analysis tools~\cite{livshits2015defense}, our learner-based bug detectors are inherently unsound, 
however, their strength lies in their far lower cost of design (\ie no human expertise is required), 
and comparable or even superior performance in practice.

\begin{figure}
    \centering
    \adjustbox{max width=.98\columnwidth}{
	\lstset{style=mystyle}
\begin{lstlisting}[numbers=left,linewidth=8.5cm,framexleftmargin=8pt]
private MavenExecutionResult doExecute(
    MavenExecutionRequest request) {
    ...
    projectDependencyGraph = createProjectDependencyGraph(
        projects,request,result,true);

    %\footnotesize{\it{session.setProjects(
        projectDependencyGraph.getSortedProjects())}}%;
    
    if (result.hasExceptions()) 
        return result;
    ...
}
\end{lstlisting}}
	\caption{A null pointer dereference bug (bugs-dot-jar\_MNG-5613\_bef7fac6). The presence of the exception handling routine at line 9-10 indicates the function at line 4 may throw an exception, in which case the value of \texttt{projectDependencyGraph} is likely to be null; subsequently, a null pointer dereference will be triggered at line 7.}
    \label{fig:1}
\end{figure}

\begin{figure}
    \centering
    \adjustbox{max width=.98\columnwidth}{
	\lstset{style=mystyle}
\begin{lstlisting}[numbers=left,linewidth=8.3cm,framexleftmargin=8pt]
public JavaElementLine(ITypeRoot element, int lineNumber, 
    int lineStartOffset) throws CoreException {
    ...
    IBuffer buffer = element.getBuffer();
    ...
    int length = buffer.getLength();
    int i = lineStartOffset;

    char ch = Array.getChar(buffer,i);
    ...
    while (i < length && 
        !IndentManipulation.isLineDelimiterChar(ch)) {
  	    ...
	      %\footnotesize{\it{ch = Array.getChar(buffer, ++i);}}%
    ...
}
\end{lstlisting}}
	\caption{An array index out of bound bug (eclipse.jdt.ui-2f5e78c). Since the pre-incrementor at line 14 increments the value of \texttt{i} right before accessing the array \texttt{buffer}, an array index out of bound exception will be triggered in the iteration where \texttt{i} equals to \texttt{length}-1.}
    \label{fig:2}
\end{figure}

\begin{figure}
    \centering
    \adjustbox{max width=.98\columnwidth}{
	\lstset{style=mystyle}
\begin{lstlisting}[numbers=left,linewidth=8.3cm,framexleftmargin=8pt]
public static void addMissingReturnTypeProposals(
    IInvocationContext context, IProblemLocation problem, 
    Collection<ICommandAccess> proposals) {
    ...
    if (parentType instanceof AbstractTypeDeclaration) {
        boolean isInterface = 
            parentType instanceof TypeDeclaration && 
            ((TypeDeclaration)parentType).isInterface();  

        if (!isInterface) {
            %\footnotesize{\it{String constructorName = ((TypeDeclaration)parentType).}}%
               %\footnotesize{\it{getName().getIdentifier()}}%;
            ...
}
\end{lstlisting}}
	\caption{An class cast bug (eclipse.jdt.ui-c0e0634). If \texttt{isInterface} (line 6) is evaluated to false, it's possible \texttt{parentType} is not an instance of \texttt{TypeDeclaration}, in which case a class cast exception will be triggered at line 11.}
    \label{fig:3}
\end{figure}

Our approach significantly differs from almost all related work in the literature~\cite{wang2016automatically,pradel2018deepbugs,allamanis2017learning}. Specifically, we consider deep and semantic bugs that are proven to be hard even for state-of-the-art static analyzers instead of shallow and syntactic bugs targeted by other works~\cite{wang2016automatically,allamanis2017learning,pradel2018deepbugs}. 
Additionally, \tool pinpoints a bug in a program to a line instead of predicting an entire file to be buggy or not~\cite{wang2016automatically}.

We evaluate \tool by training its three instantiations on 30,000 methods extracted from hundreds of Java projects and detecting bugs in 1,700 methods from a different set of 13 Java projects. 
In total, our corpus amounts to 1,047,296 lines. We find that all neural bug detectors are effective with each achieving above 35\% precision and 45\% recall. 

To further demonstrate the utility of \tool, we compare our neural bug detector with several static analysis tools (\eg Facebook's Infer~\cite{calcagno2015moving,berdine2005smallfoot}, Pinpoint~\cite{shi2018pinpoint}, \etc). Our comparison focuses on the performance of each checker in catching null pointer dereference bugs, the only kind of bugs that can be handled by all tools. Results show that our neural bug detector catches more bugs and produces less spurious warnings than any evaluated static analyzer. We have also applied our neural bug detectors to 17 popular projects on GitHub and discovered 50 new bugs, among which 9 have been fixed and 3 have been confirmed (fixes pending).

We make the following contributions:
\begin{itemize}
    \item We propose a deep neural network based methodology for building static bug checkers. Specifically, we utilize GNN to train a classifier for differentiating buggy code from correct code.
    \item  We propose a novel interval-based propagation model that significantly enhances the capacity of GNN in learning programming patterns.
    \item We design and implement a framework to streamline the creation of neural bug detectors
    The framework is open-sourced at \url{https://github.com/anonymoustool/NeurSA}.
    \item We publish our data set at \url{https://figshare.com/articles/datasets_tar_gz/8796677} for the three neural bug detectors we built based on \tool to aid the future research activity.
    \item We present the evaluation results showing our neural bug detectors are highly precise in detecting the semantic bugs in real-world programs and outperform many flagship static analysis tools including Infer in catching null pointer dereference bugs. 
\end{itemize}

\section{Preliminary}\label{sec:Pre}

First, we revisit the definition of connected, directed graphs. Then we give a brief overview of interval~\cite{Allen:1970:CFA:800028.808479} and GNN~\cite{gori2005new}, which our work builds on.

\subsection{Graph}
A graph $\mathcal{G}\! =\! (\mathcal{V}, \mathcal{E})$ consists of a set of nodes $\mathcal{V} \!= \!\{v_1, ..., v_m\}$, and a list of directed edge sets $\mathcal{E} = (\mathcal{E}_1, . . . , \mathcal{E}_K)$ where $K$ is the total number of edge types and $\mathcal{E}_k$ is a set of edges of type $k$. $E(v_s, v_d, k), k\in(1,...,K)$ denotes an edge of type $k$ directed from node $v_s$ to node $v_d$. For graphs with only one edge type, $E$ is represented as $(v_s, v_d)$.

The immediate successors of a node $v_i$ (denoted as $\textit{post}\,(v_i)$) are all of the nodes $v_j$ for which $(v_i, v_j)$ is an edge in $\mathcal{E}$. The immediate predecessors of node $v_j$  (denoted as $\textit{pre}\,(v_j)$) are all of the nodes $v_i$ for which $(v_i, v_j)$ is an edge in $\mathcal{E}$. 

A path is an ordered sequence of nodes $(v_j,..., v_k)$ and their connecting edges, in which each $v_i$ $\in$ $\textit{pre}\,(v_{i+1})$ for $i \in (j,...,k-1)$. A closed path is a path in which the first and last nodes are the same. The successors of a node $v_i$ (denoted as $\textit{post}\,^{*}(v_i)$) are all of the nodes $v_j$ for which there exists a path from $v_i$ to $v_j$. The predecessors of a node $v_j$  (denoted as $\textit{pre}\,^{*}(v_j)$) are all of the nodes $v_i$ for which there exists a path from $v_i$ to $v_j$.

\subsection{Interval}


Introduced by~\citet{Allen:1970:CFA:800028.808479},
an interval \textit{I(h)} is the maximal,
single entry subgraph in which \textit{h} is the only entry node
and all closed paths contain \textit{h}. The unique
interval node \textit{h} is called the interval head or simply the
header node. An interval can be expressed in terms of
the nodes in it: $I(h) = \{v_l, v_2, ... ,v_m\}. $

By selecting the proper set of header nodes, a graph can be partitioned into a set of disjoint intervals. An algorithm for such a partition is shown in Algorithm \ref{alg:intervals}. The key is to add to an interval a node only if all of whose immediate predecessors are already in the interval (Line \ref{li:w2} to \ref{li:s2}). The intuition is such nodes when added to an interval keep the original header node as the single entry of an interval. To find a header node to form another interval, a node is picked that is not a member of any existing intervals although it must have a (not all) immediate predecessor being a member of the interval that is just computed (Line \ref{li:w3} to \ref{li:s3}). We repeat the computation until reaching the fixed-point where all nodes are members of an interval.
 

\begin{algorithm}[tbp]
\small
\caption{Finding intervals for a given graph} 
\label{alg:intervals}
\raggedright
\KwIn {Graph $\mathcal{G}$, Node set $\mathcal{V}$}
\KwOut {Interval set $\mathcal{S}$}
// $v_0$ is the unique entry node for the graph

H = \{$v_0$\}\;
\While{$H \neq \emptyset$}
{
        // remove next $h$ from $H$
        
        h = H.pop()\;
        I(h) = \{h\}\;
        // only nodes that are neither in the current interval nor any other interval will be considered

        \While{$\{v \in \mathcal{V} \,| \, v \notin I(h) \wedge \nexists s (s \in \mathcal{S} \wedge v \in s) \wedge  pre(v) \subseteq I(h)\} \neq \emptyset$\label{li:w2}}  
        {
           I(h) = I(h) $\cup$ \{ v \}\; \label{li:s2}
        }
        // find next headers
        
        \While{$\{v \in \mathcal{V}  \,| \, \nexists s_{1} (s_{1} \in \mathcal{S} \wedge v \in s_{1}) \, \wedge\!$ $\exists m_{1}, m_{2}$ $(m_{1}$ $\in$ $pre(v)$ $\wedge$ $m_{2} \in pre(v) \wedge m_{1} \in I(h) \wedge m_{2} \notin I(h)) \} \neq \emptyset \; \; \;$ \label{li:w3}}
        {
            H = H $\cup$ \{ v \}\; \label{li:s3}
        }
        $\mathcal{S}$ = $\mathcal{S}$ $\cup$ I(h)\;
}
\end{algorithm}

The intervals on the original graph 
are called the first order intervals denoted by 
$I^1(h)$, and the graph from which they were derived 
is called first order graph (also called the set 
of first order intervals) $\mathcal{S}^1$ \st 
$I^1(h) \in \mathcal{S}^1$. By making each first order 
interval into a node and each interval exit edge into 
an edge, the second order graph can be derived, from 
which the second order intervals can also be defined. 
The procedure can be repeated to derive successively 
higher order graphs until the n-th order graph consists 
of a single node\footnote{Certain graphs can't be reduced to single nodes.}. Figure~\ref{fig:interval} illustrates 
such a sequence of derived graphs. 





\begin{figure}[tbp]
\centering
  \includegraphics[width=1\linewidth]{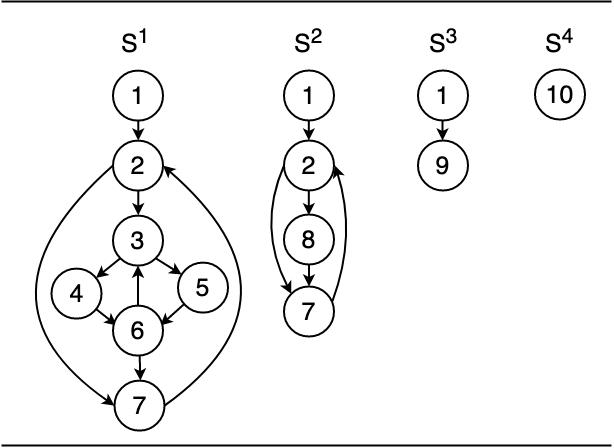}
  \caption{n-th order intervals and graphs. The set of intervals on $\mathcal{S}^1$ are $I^1(1)$=\{1\}, $I^1(2)$=\{2\}, $I^1(3)$=\{3,4,5,6\}, $I^1(7)$=\{7\}. $I^2(1)$=\{1\} and $I^2(2)$=\{2,7,8\} are the second order intervals. $I^3(1)$=\{1,9\} and $I^4(10)$=\{10\} are the only intervals on $\mathcal{S}^3$ and $\mathcal{S}^4$ respectively.}
  \label{fig:interval}
\end{figure}

\subsection{Graph Neural Network}
Graph Neural Networks (GNN)~\cite{scarselli2008graph,gori2005new} is a specialized machine learning model designed to learn from graph data.

We extend the definition of a graph to include $\mathcal{M}$ (\ie $G = (\mathcal{V}, \mathcal{E}, \mathcal{M})$). $\mathcal{M}$ is a set of vectors (or embeddings) $\{\mu_{v_1}, ..., \mu_{v_m}\}$, where each $\mu_{v} \in \mathbb{R}^d$ denotes the embedding of a node $v$ in the graph. GNN updates node embeddings via a propagation model. The simplest yet most popular is synchronous message passing systems~\cite{attiya2004distributed} introduced in distributed computing theory. Specifically, the inference is executed as a sequence of rounds: in each round, every node first sends messages to all of its neighbors, and then update its embedding by aggregating all incoming messages.
\begin{equation}
\mu^{(l+1)}_{v} = h(\{\mu^{(l)}_{u}\}_{u \in \mathcal{N}^{k}(v), k \in \{1,2,...,K\}}) \label{equ:h}
\end{equation}

$\mathcal{N}^{k}{(v)}$ denotes the neighbours that are connected to $v$ with edge type $k$, \ie, $\mathcal{N}^{k}{(v)} = \{u|(u,v,k) \in \mathcal{E}\} \cup \{u|(v,u,k) \in \mathcal{E}\} $. $h(\cdot)$ denotes the aggregation function. We repeat the propagation for $L$ steps to update $\mu_{v}$ to $\mu^{(L)}_{v}, \forall v \in \mathcal{V}$. 

Some GNNs~\cite{si2018learning} compute a separate node embedding \wrt an edge type (\ie $\mu^{(l+1), k}_{v}$ in Equation~\ref{equ:type}) before aggregating them into a final embedding (\ie $\mu^{(l+1)}_{v}$ in Equation~\ref{equ:sum}).
\begin{align}
\mu^{(l+1), k}_{v} &= \phi_1(\sum_{\mathclap{u \in \mathcal{N}^{k}{(v)}}} \mathbf{W_1} \mu_u^{(l)}), \forall k \in \{1,2,...,K\} \label{equ:type}\\
\mu^{(l+1)}_{v} &= \phi_2(\mathbf{W_2}[ \mu_v^{(l), 1}, \mu_v^{(l), 2}, ..., \mu_v^{(l), K}]) \label{equ:sum}
\end{align}

$\mathbf{W_1}$ and $\mathbf{W_2}$ are variables to be learned, and $\phi_1$ and $\phi_2$ are some nonlinear activation functions.

To further improve the model capacity,~\citet{li2015gated} proposed Gated Graph Neural Network (GGNN). Their major contribution is a new instantiation of $h(\cdot)$ (Equation~\ref{equ:h}) with Gated Recurrent Units~\cite{cho2014properties}. The following equations describe how GGNN works:
\begin{align}
\tilde{m}_v^l &= \sum_{\mathclap{u \in \mathcal{N}{(v)}}} f(\mu_{u}^{(l)}) \label{equ:mes1}\\
\mu^{(l+1)}_{v} &= GRU(\tilde{m}_v^l, \mu^{(l)}_{v})\label{equ:mes2}
\end{align}

To update the embedding of node $v$, Equation~\ref{equ:mes1} computes a message $\tilde{m}_v^l$ using $f(\cdot)$ (\eg a linear function) from the embeddings of its neighboring nodes $\mathcal{N}{(v)}$. Next a $GRU$ takes $\tilde{m}_v^l$ and $\mu^{(l)}_{v}$---the current embedding of node $v$---to compute the new embedding, $\mu^{(l+1)}_{v}$ (Equation~\ref{equ:mes2}).

\section{Interval-Based Propagation Mechanism}
\label{sec:ibpm}

In this section, we present the Interval-Based Propagation Mechanism (IBPM), a new protocol that regulates how nodes exchange messages with their peers in a graph. In particular, we highlight the conceptual advantages IBPM enjoys over the existing propagation mechanism.

\subsection{IBPM's Algorithm for Graph Propagation}
\label{subsec:IBPM}

We discuss two major issues existing works that apply GNN in program analysis suffer from. 

Even though GNN has shown its generality by succeeding in a variety of problem domains (\eg scene graph generation~\cite{8099813}, traffic prediction~\cite{li2017diffusion}, recommendation system~\cite{Ying:2018}, \etc), it can benefit from specializations designed to address the uniqueness of each problem setting, especially when dealing with programs, a fundamentally different data structure. However, prior works made little adjustment to accommodate such program-specific graph characteristics, resulted in a potential loss of model precision. Besides, scalability challenges can also hinder the generalization of GNN. That is when a graph has a large diameter; information has to be propagated over long distance, therefore, message exchange between nodes that are far apart becomes difficult. 

IBPM is designed specifically to tackle the aforementioned weaknesses of GNN. To start with, programs are represented by their control flow graphs. More distinctively, propagation is only allowed on a subgraph defined by an interval. Since iteration statements (\ie loops) often determines the formation of intervals (because loop back edge \eg node 7 to node 2 in Figure~\ref{fig:interval} cuts off the loop header into a new interval), IBPM, in essence, attends to such loop structures when propagating on a control-flow graph. To enable the message-exchange between nodes that are far apart, IBPM switches onto a higher order graph, as such, involving more nodes to communicate within the same interval. On the other hand, by reducing the order of graph, IBPM restores the location propagation, and eventually recovers the node embeddings on the original control flow graph. 
Below we use the program in Figure~\ref{fig:ibpm}, a function that computes the starting indices of the substring \texttt{s1} in string \texttt{s2}, and its control-flow graph in Figure~\ref{fig:interval} to explain how IBPM works.

\lstset{style=mystyle}
\begin{figure*}
        \centering
        \begin{tikzpicture}
        \matrix[matrix of nodes, 
        nodes={anchor=center},
        column sep=-1em,
        row sep = -.5em]{%
        {%
        \setlength{\fboxsep}{1.5pt}%
        \setlength{\fboxrule}{.5pt}%
        \fbox{\hspace{3pt}\includegraphics[width=0.28\textwidth]{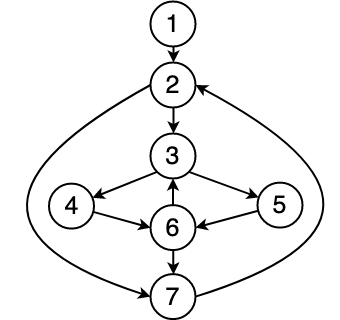}\hspace{6.1pt}}%
        }%
            &
            \thinspace\thinspace
            \includegraphics[width=0.03\textwidth]{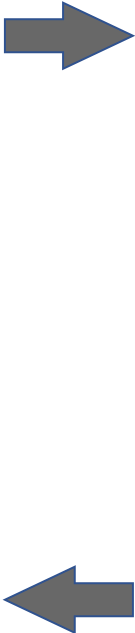}
            &\thinspace\thinspace
            {%
                \setlength{\fboxsep}{2.6pt}%
                \setlength{\fboxrule}{.5pt}%
                \fbox{\hspace{3.5pt}\includegraphics[width=0.28\textwidth]{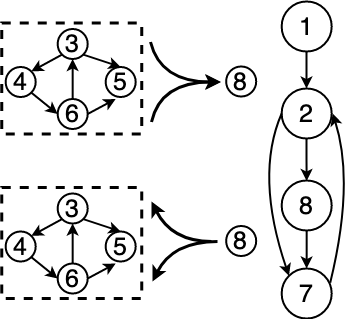}\hspace{5.2pt}}%
             }            
             &
            \thinspace\thinspace
            \includegraphics[width=0.03\textwidth]{Figures/Picture1.png}
            \thinspace
            &
            {%
                \setlength{\fboxsep}{5pt}%
                \setlength{\fboxrule}{.5pt}%
                \fbox{\hspace{1.5pt}\includegraphics[width=0.28\textwidth]{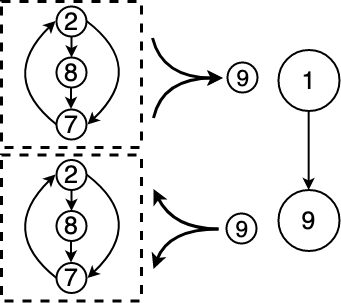}\hspace{2pt}}%
             }
             \\
            \begin{minipage}[t][8ex][t]{0.28\textwidth}
            \subcaption{Propagation among node 3, 4, 5 and 6 on the first order graph.\label{fig:fog}}
            \end{minipage}
            &&
            \begin{minipage}[t][8ex][t]{0.3\textwidth}
            \subcaption{Propagation among node 2, 7 and 8 on the second order graph.\label{fig:sog}}
            \end{minipage}%
            &&
            \begin{minipage}[t][8ex][t]{0.3\textwidth}
            \subcaption{Propagation among node 1 and 9 on the third order graph.\label{fig:tog}}
            \end{minipage}%
            \\    
            \vspace{-1cm}
            \begin{minipage}[t]{0.28\textwidth}
\begin{lstlisting}[linewidth=5.1cm,xleftmargin=-0.4em,xrightmargin=0em,framexleftmargin=0em]
void substringIndices(%\color{blue}{String}% s1, 
	%\color{blue}{String}% s2, int attempts, int[] arr) 
{
	%\color{light-gray}{int M = s1.length();}%
	%\color{light-gray}{int N = s2.length();}%
	%\color{light-gray}{int i = 0,j = 0,res = 0,pos = 0;}%			

	%\color{light-gray}{do \{}%
		%\color{light-gray}{if (pos < attempts) \{}%
			do {
				log("Begin at index: " + i);
				if (s2.charAt(i+j) != 
						s1.charAt(j)) {
					res =  -1;
				  log("Differs at: " + j);}
				else
					log("Equal at: " + j);
				j++;
			} while(j < M);
    %\color{light-gray}{\}}%
		
		%\color{light-gray}{arr[pos++] = res;}%
		%\color{light-gray}{i++; j = 0;}%
		%\color{light-gray}{res = pos < attempts ? i : 0;}%		
	%\color{light-gray}{\} while (i <= N - M);}%
%\color{light-gray}{\}}%
\end{lstlisting}
            \end{minipage}
            &
            \thinspace\thinspace
            \includegraphics[width=0.03\textwidth]{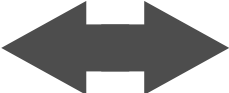}
            &
            \begin{minipage}[t]{0.28\textwidth}
\begin{lstlisting}[linewidth=5.4cm,xleftmargin=0.55em,xrightmargin=0em,framexleftmargin=0em]
void substringIndices(%\color{blue}{String}% s1, 
	%\color{blue}{String}% s2, int attempts, int[] arr) 
{
	%\color{light-gray}{int M = s1.length();}%
	%\color{light-gray}{int N = s2.length();}%
	%\color{light-gray}{int i = 0,j = 0,res = 0,pos = 0;}%
	
	do {
		if (pos < attempts) {
			do {
				log("Begin at index: " + i);
				if (s2.charAt(i+j) != 
						s1.charAt(j)) {
					res =  -1;
				  log("Differs at: " + j);}
				else
					log("Equal at: " + j);
				j++;
			} while(j < M);
    }
		
		arr[pos++] = res;
		i++; j = 0;
		res = pos < attempts ? i : 0;
	} while (i <= N - M);
%\color{light-gray}{\}}%
\end{lstlisting}
            \end{minipage}
            &
            \includegraphics[width=0.03\textwidth]{Figures/arrowd.png}
            &
            \begin{minipage}[t]{0.28\textwidth}
\begin{lstlisting}[linewidth=5.3cm,xleftmargin=0.15em,xrightmargin=0em,framexleftmargin=0em]
void substringIndices(%\color{blue}{String}% s1, 
	%\color{blue}{String}% s2, int attempts, int[] arr) 
{
	int M = s1.length();
	int N = s2.length();
	int i = 0,j = 0,res = 0,pos = 0;
	
	do {
		if (pos < attempts) {
			do {
				log("Begin at index: " + i);
				if (s2.charAt(i+j) != 
						s1.charAt(j)) {
					res =  -1;
				  log("Differs at: " + j);}
				else
					log("Equal at: " + j);
				j++;
			} while(j < M);
    }
		
		arr[pos++] = res;
		i++; j = 0;
    res = pos < attempts ? i : 0;
	} while (i <= N - M);
}
\end{lstlisting}
            \end{minipage}
            \\
            \begin{minipage}[t][7ex][t]{0.28\textwidth}
            \subcaption{Illustration of IBPM on the first order graph using example program.\label{fig:tran1}}
            \end{minipage}
            &&
            \begin{minipage}[t][7ex][t]{0.28\textwidth}
            \subcaption{Illustration of IBPM on the second order graph using example program.\label{fig:tran2}}
            \end{minipage}%
            &&
            \begin{minipage}[t][7ex][t]{0.28\textwidth}
            \subcaption{Illustration of IBPM on the third order graph using example program.\label{fig:tran3}}
            \end{minipage}%
            \\
        };
        \end{tikzpicture}
        \caption{Interval-Based Propagation Mechanism.\label{fig:ibpm}}
    \end{figure*}

Depicted in Figure~\ref{fig:fog}, IBPM starts with the first order graph (\ie the original control flow graph). As explained earlier, we restrict communication among nodes within the same interval. Consequently, node 3, 4, 5 and 6 will freely pass messages to the neighbors along their directed edges. On the contrary, node 1 and 2, which are alone in their respective intervals, will not contact any other node in the graph. Such propagation that occurs internal to an interval is carried out in the same manner as the standard propagation (Equation~\ref{equ:h}). Figure~\ref{fig:tran1} highlights the part of the example program that corresponds to the only active interval in the current step of IBPM. Apparently, IBPM focuses on learning from the inner loop while leaving the rest of the program unattended. 

We then move onto the second order graph as shown in Figure~\ref{fig:sog}. Recall how a second order graph is derived. We compute a new node (\ie node 8) to replace the only active interval on the first order graph (\ie node 3, 4, 5 and 6). Equation~\ref{equ:merge} defines how to initialize node 8 or any other node that is created when IBPM transitions from lower-order to higher-order graphs.
\begin{equation}\label{equ:merge}
\mu^{(0)}_{w} = \sum_{\mathclap{v \in I^{i}(h)}} \alpha_{v}\mu^{(l)}_{v}
\end{equation}

where $w$ denotes the node created out of the interval $I^{i}(h)$ and $\alpha_{v}$ is defined by 
\begin{equation}\label{equ:weight}
\alpha_{v} = \frac{exp\,(\mu^{(l)}_{v})}{\sum\limits_{\mathclap{v \in I^{i}(h)}} exp\,(\mu^{(l)}_{v})}         
\end{equation}


Due to the same restriction IBPM imposes, propagation on the second order graph now covers node 2, 7 and 8. Because node 8 can be deemed as a proxy of node 3-6, this propagation, in essence, covers all but node 1 on the first order graph. In other words, even though propagation is always performed within each interval, increasing the order of a graph implicitly expands the region of an interval, therefore facilitating more nodes to communicate. For visualization, Figure~\ref{fig:tran2} highlights the part of the program---the nested loop structure---IBPM attends to at the current step.

Figure~\ref{fig:tog} shows the third step of IBPM, where propagation spans across the entire graph. Upon completion of the message-exchange between node 1 and 9, IBPM transitions back to the lower order graphs; interval nodes will be split back into the multitude of nodes from which they formed (\eg node 9 to node 2, 7 and 8; node 8 to node 3-6). Equation below defines how to compute the initial embeddings for a node that is split from an interval. $\mu_{w}$ and $\alpha_{v}$ denote the same meaning in Equation~\ref{equ:merge}.
\begin{equation*}\label{equ:split}
\mu^{(0)}_{v} = \alpha_{v}\mu^{(l)}_{w}
\end{equation*}

Eventually, the node embeddings on the original control flow graph will be recovered. Worth mentioning recurring local propagation on the lower order graphs improves the accuracy of node embeddings by incorporating the previously accumulated global graph knowledge on the higher order graphs. In return, more precise modeling of the subgraphs also benefits the further global propagation across the entire graph. As the transitioning between local and global mode repeats, node embeddings will continue to be refined to reflect the properties of a graph.

\subsection{Comparing IBPM to Standard Propagation}

\paragraph{\textbf{Structured Propagation}}
As described in Section~\ref{subsec:IBPM}, IBPM works with a graph hierarchy, where each level attends to distinct, important code constructs in a program, such as the loops illustrated in Figure~\ref{fig:tran1} and~\ref{fig:tran2}. This is in stark contrast with the standard propagation mechanism, where program structure is largely ignored and messages spread across the entire graph. Intuitively, we believe by utilizing the program constructs to regulate the message exchange, IBPM improves the precision of GNN in learning program properties.

\paragraph{\textbf{Scalable Propagation}}
We prove in theory IBPM also enjoys scalability advantage over the standard propagation mechanism. By scalability, we mean the number of messages that have been passed in total until the graph is sufficiently propagated. Since determining the threshold of sufficient propagation is beyond the scope of this paper, we make reasonable approximations to assist our proof. For simplicity, we assume directed, connected graphs which consist of distinct nodes and homogeneous edges, and synchronous message passing scheme in all definitions, theorems, and proofs.  

\begin{definition}{(Distance)}
\label{def:dis}
Given two nodes $v_{i}$ and $v_{j}$ on a graph $\mathcal{G}$, their distance, denoted by $\theta(v_{i},v_{j})$, is the length of the shortest path from $v_{i}$ to $v_{j}$ on $\mathcal{G}$.
\end{definition}

\begin{definition}{(Diameter)}
\label{def:dia}
The diameter of a graph $\mathcal{G}$, $\pi(\mathcal{G})$, is the longest distance between any two nodes on $\mathcal{G}$. Assuming $\mathcal{V}$ is the set of nodes on $\mathcal{G}$.
\end{definition}
\begin{equation*}
\pi(\mathcal{G}) = \!\!\!\maxx_{v_{i},v_{j} \in \mathcal{V}} \! \! \!  \theta(v_{i},v_{j})        
\end{equation*}

\noindent
We define the point of sufficient propagation as follows.

\begin{definition}{(Fixed Point)}
\label{def:fp}
Given a graph $\mathcal{G}\! =\! (\mathcal{V}, \mathcal{E})$, associate each node $v \in \mathcal{V}$ with a set $\sigma_{v}$, which is initialized to contain the node itself, denoted by $\sigma_{v}^{0} = \{v\}$. After the $l$-th round of message passing, $\sigma_{v}^{l}$ is updated to $\bigcup_{u \in \mathcal{N}(v)} \sigma_{u}^{l-1}$, where $\mathcal{N}(v)$ refers to the neighbouring nodes of $v$. The propagation is said to reach the fixed point on $\mathcal{G}$ after $l$-th round of message passing iff (1) $\forall l^{\prime}\geq l,\, \forall v \in \mathcal{V},\, \sigma_{v}^{l} = \sigma_{v}^{l^{\prime}}$; and (2) $\nexists l^{\prime}<l,\, \forall l^{\prime\prime}\geq l^{\prime},\, \forall v \in \mathcal{V},\, \sigma_{v}^{l^{\prime}} = \sigma_{v}^{l^{\prime\prime}}$.
\end{definition}



\begin{theorem}[Distance Measure]
\label{the:dm}
Given two nodes $v_{i}$ and $v_{j}$ on a graph (hereinafter assuming all edges are of unit length), if $\sigma_{v_{j}}$ contains $v_{i}$ only after the $l$-th round of message passing, that is $\,\,\forall l^{\prime}<l,\, v_{i} \notin \sigma_{v_{j}}^{l^{\prime}}$ and $\, \forall l^{\prime}\geq l,\, v_{i} \in \sigma_{v_{j}}^{l^{\prime}}$, then $\theta(v_{i},v_{j}) = l$.
\end{theorem}

\begin{proof}
Assume otherwise \st $\theta(v_{i},v_{j}) = l^{\prime} \wedge l \neq l^{\prime}$. Synchronous message passing system in essence enumerates all possible paths between any two nodes on the graph. In particular, all paths of length $l$ will be enumerated in $l$-th round of message passing. If $\theta(v_{i},v_{j}) = l^{\prime}$, then $l^{\prime}$ rounds would have been needed, after which $\sigma_{v_{j}}$ contains $v_{i}$, which contradicts the assumption $\sigma_{v_{j}}$ contains $v_{i}$ after the $l$-th round of message passing.
\end{proof}

\begin{theorem}[Scalability Measure]
\label{the:sm}
Given a graph $\mathcal{G}\! =\! (\mathcal{V}, \mathcal{E})$, the number of messages that have been sent by all nodes in $\mathcal{V}$ until the propagation reaches fixed point on $\mathcal{G}$ equals to $\pi(\mathcal{G})\, * \lvert \mathcal{E} \rvert$.
\end{theorem}

\begin{proof}
Because synchronous message passing system has $\lvert \mathcal{E} \rvert$ message sent in each round on directed graphs. So it suffices to prove that $\pi(\mathcal{G})$ rounds of message passing are needed until propagation reaches the fixed point. Assume otherwise, reaching the fixed point on $\mathcal{G}$ takes $l^{\prime}$ rounds:

(a)  $l^{\prime} > \pi(\mathcal{G})$. This indicates $\exists v_{i} \in \mathcal{V}, \, \exists v_{j} \in \mathcal{V}, \, \forall l^{\prime\prime}<l^{\prime},\, v_{i} \notin \sigma_{v_{j}}^{l^{\prime\prime}} \wedge v_{i} \in \sigma_{v_{j}}^{l^{\prime}}$. According to Theorem~\ref{the:dm}, then $\theta(v_{i},v_{j}) = l^{\prime}$. Since $l^{\prime} > \pi(\mathcal{G})$, therefore contradicting the Definition~\ref{def:dia}. 

(b) $l^{\prime} < \pi(\mathcal{G})$. This indicates $\nexists v_{i} \in \mathcal{V}, \nexists v_{j} \in \mathcal{V}, \nexists l^{\prime\prime} > l^{\prime}, \forall l^{\prime\prime\prime}<l^{\prime\prime}, v_{i} \notin \sigma_{v_{j}}^{l^{\prime\prime\prime}} \wedge v_{i} \in \sigma_{v_{j}}^{l^{\prime\prime}}$. According to Theorem~\ref{the:dm}, then $\nexists v_{i} \in \mathcal{V}, \nexists v_{j} \in \mathcal{V}, \nexists l^{\prime\prime} > l^{\prime}, \theta(v_{i},v_{j}) = l^{\prime\prime}$, and therefore $l^{\prime} \geq \pi(\mathcal{G})$, which contradicts the assumption $l^{\prime} < \pi(\mathcal{G})$.
\end{proof}

\noindent
Similarly, we define the saturation point of IBPM.

\begin{definition}{(Fixed Point for IBPM)}
\label{def:fpI}
Given the first order graph $\mathcal{S}^1$, and successively derived higher order graphs, $\mathcal{S}^2,\mathcal{S}^3,...,\mathcal{S}^n$ ($\mathcal{S}^n$ represents the highest order graph), propagation transitions from $\mathcal{S}^i$ to $\mathcal{S}^{i+1}$ (or $\mathcal{S}^{i-1}$) only after the fixed point on all intervals in $\mathcal{S}^i$ has been reached (according to the Definition~\ref{def:fp}). In addition, $\sigma_{v^{i+1}}^{0}$ of a new node $v^{i+1}$ on $\mathcal{S}^{i+1}$ is initialized to $\bigcup_{u \in I^{i}(v)} \sigma_{u}^{l}$, where $I^{i}(v)$ denotes the interval from which $v^{i+1}$ is created, and $l=\pi(I^{i}(v))$. As for a new node $v^{i-1}$ on $\mathcal{S}^{i-1}$, $\sigma_{v^{i-1}}^{0}$ is initialized to $\{v^{i-1}\}$. The whole process is said to reach fixed point iff the propagation successively transitions from $\mathcal{S}^{1}$ to $\mathcal{S}^{n}$ and back to $\mathcal{S}^{1}$, on which subsequent propagation also reaches fixed point.
\end{definition}

\begin{theorem}[Scalability Measure for IBPM]
\label{the:smI}
Given $\mathcal{S}^1,\mathcal{S}^2$, $...,\mathcal{S}^n$ and their respective member intervals $I^1(h_{1}^{1}),...,I^1(h_{\,\rule{.4pt}{1.2ex} \mathcal{S}^{1} \rule{.4pt}{1.2ex}}^{1})$; $I^2(h_{1}^{2}),...,I^2(h_{\,\rule{.4pt}{1.2ex} \mathcal{S}^{2} \rule{.4pt}{1.2ex}}^{2})$; $I^n(h_{1}^n),...,I^n(h_{\,\rule{.4pt}{1.2ex} \mathcal{S}^{n} \rule{.4pt}{1.2ex}}^{n})$, the number of messages that have been sent until IBPM reaches fixed point equals to 
\begin{equation}\label{equ:ibpmsca}
\bigg( \sum_{j=1}^{n-2} \sum_{i=1}^{\lvert \mathcal{S}^{j} \! \rvert} 2 * \pi(I^{j}(h^{j}_{i})) * \lvert \mathcal{E}_{I^{j}(h^{j}_{i})} \rvert \bigg) + \pi(\mathcal{S}^{n-1}) * \lvert \mathcal{E}_{\mathcal{S}^{n-1}} \rvert
\end{equation}

\noindent if $\mathcal{S}^{n}$ is a single node, otherwise
\begin{equation*}
\!\!\!\!\!\!\!\!\!\!\!\!\!\!\bigg( \sum_{j=1}^{n-1} \sum_{i=1}^{\lvert \mathcal{S}^{j} \! \rvert} 2 * \pi(I^{j}(h^{j}_{i})) * \lvert \mathcal{E}_{I^{j}(h^{j}_{i})} \rvert \bigg)  + \pi(\mathcal{S}^{n}) * \lvert \mathcal{E}_{\mathcal{S}^{n}} \rvert
\end{equation*}
\noindent $\mathcal{E}_{g}$ denotes the set of edges in $g$ in both formulas.

\end{theorem}
\begin{proof}
Break IBPM into sub-propagation on each interval at each order of graph. Therefore, the number of messages sent in total can be computed by summing up the messages sent during each sub-propagation. Since all assumptions still hold,
Theorem~\ref{the:sm} can be applied to derive the proof.
\end{proof}

Below we compare the scalability of two propagation mechanism. First, we consider the case where $\mathcal{S}^{n}$ is a single node. Add another $\pi(\mathcal{S}^{n-1}) * \lvert \mathcal{E}_{\mathcal{S}^{n-1}} \rvert$ to Equation~\ref{equ:ibpmsca} leads to:
\begin{equation} \label{equ:deri1}
\setlength{\jot}{1pt} 
\begin{split}
&\!\!\!\!\!\bigg( \sum_{j=1}^{n-2} \sum_{i=1}^{\lvert \mathcal{S}^{j} \! \rvert} 2 * \pi(I^{j}(h^{j}_{i})) * \lvert \mathcal{E}_{I^{j}(h^{j}_{i})} \rvert \bigg) + \pi(\mathcal{S}^{n-1}) * \lvert \mathcal{E}_{\mathcal{S}^{n-1}} \rvert < \\
&\!\!\!\!\!\bigg( \sum_{j=1}^{n-2} \sum_{i=1}^{\lvert \mathcal{S}^{j} \! \rvert} 2 * \pi(I^{j}(h^{j}_{i})) * \lvert \mathcal{E}_{I^{j}(h^{j}_{i})} \rvert \bigg)  + 2 * \pi(\mathcal{S}^{n-1}) * \lvert \mathcal{E}_{\mathcal{S}^{n-1}} \rvert
\end{split}
\end{equation}

Absorb $2 * \pi(\mathcal{S}^{n-1}) * \lvert \mathcal{E}_{\mathcal{S}^{n-1}} \rvert$ into the sum at the right hand side of Equation~\ref{equ:deri1} gives
\begin{equation} \label{equ:deri2}
\setlength{\jot}{1pt} 
\begin{split}
\!\!\!\!\!&\bigg( \sum_{j=1}^{n-2} \sum_{i=1}^{\lvert \mathcal{S}^{j} \! \rvert} 2 * \pi(I^{j}(h^{j}_{i})) * \lvert \mathcal{E}_{I^{j}(h^{j}_{i})} \rvert \bigg) + \pi(\mathcal{S}^{n-1}) * \lvert \mathcal{E}_{\mathcal{S}^{n-1}} \rvert < \\
\!\!\!\!\!&\sum_{j=1}^{n-1} \sum_{i=1}^{\lvert \mathcal{S}^{j} \! \rvert} 2 * \pi(I^{j}(h^{j}_{i})) * \lvert \mathcal{E}_{I^{j}(h^{j}_{i})} \rvert
\end{split}
\end{equation}

To simply, assume 
\[
\!\!\!\!\!\!\!\!\!\!\!\!\!\!\!\!\!\!\!\!\!\!\!\!\!\!\!\!\!\!\!\!\!\!\!\!\!\!\!\!\!\!\forall j \in [1,n-1], \forall i \in [1,\lvert \mathcal{S}^{j} \rvert], \pi(I^{j}(h^{j}_{i}))=\tau 
\]

then
\begin{equation}\label{equ:deri3}
\!\!\!\!\!\sum_{j=1}^{n-1} \sum_{i=1}^{\lvert \mathcal{S}^{j} \! \rvert} 2 * \pi(I^{j}(h^{j}_{i})) * \lvert \mathcal{E}_{I^{j}(h^{j}_{i})} \rvert
=
2 * \tau * \sum_{j=1}^{n-1} \sum_{i=1}^{\lvert \mathcal{S}^{j} \! \rvert} \lvert \mathcal{E}_{I^{j}(h^{j}_{i})} \rvert    
\end{equation}

Because $\sum_{j=1}^{n-1} \sum_{i=1}^{\lvert \mathcal{S}^{j} \! \rvert} \lvert \mathcal{E}_{I^{j}(h^{j}_{i})} \rvert \leq \lvert \mathcal{E} \rvert$. In particular, the equivalence establishes when none of the intervals at any order of graph that has multiple exit edges to the same external node. So we derive

\begin{equation} \label{equ:deri4}
\!\!\!\!\!\!\!\!\!\!\!\!\!\!\!\!\!\!\!\!\!\!\!\!\!\!\!\!\!\!\!\!\!\!\!\!\!\!\!\!\!\!\!\!\!\!\!\!\!\!\!\!\!\!\!\!\!2 * \tau * \sum_{j=1}^{n-1} \sum_{i=1}^{\lvert \mathcal{S}^{j} \! \rvert} \lvert \mathcal{E}_{I^{j}(h^{j}_{i})} \rvert
\leq
2 * \tau * \lvert \mathcal{E} \rvert
\end{equation}

Linking Equation~\ref{equ:deri1}-\ref{equ:deri4}, we derive 

\begin{equation*}
\setlength{\jot}{2pt} 
\begin{split}
\!\!\!\!\!\!\!\!\!&\bigg( \sum_{j=1}^{n-2} \sum_{i=1}^{\lvert \mathcal{S}^{j} \! \rvert} 2 * \pi(I^{j}(h^{j}_{i})) * \lvert \mathcal{E}_{I^{j}(h^{j}_{i})} \rvert \bigg) + \pi(\mathcal{S}^{n-1}) * \lvert \mathcal{E}_{\mathcal{S}^{n-1}} \rvert < \\
\!\!\!\!\!\!\!\!\!&2\tau * \lvert \mathcal{E} \rvert
\end{split}
\end{equation*}
Essentially, we have estimated an upper bound of the scalability measure for IBPM: $2\tau * \lvert \mathcal{E} \rvert$. The same upper bound can also be derived for the other case where $\mathcal{S}^{n}$ is not a single node. Comparing with Theorem~\ref{the:sm}'s result: $\pi(\mathcal{G}) * \lvert \mathcal{E} \rvert$, IBPM is more scalable than the standard propagation mechanism since the diameter of an interval is generally small (\ie around 2 on control flow graphs according to our large-scale evaluation). When the size of a graph increases, the size of an interval mostly stays as a constant, therefore IBPM outperforms the standard propagation by an even wider margin. Overall, using sensible approximations, we have proved IBPM is superior to the standard propagation in the scalability regard.




\section{Framework}\label{sec:AD}
This section presents \tool, a framework for creating various kinds of neural bug detectors. 

\subsection{Overview of \tool}
\label{subsec:over}

\begin{figure}
\centering
  \includegraphics[width=0.95\linewidth]{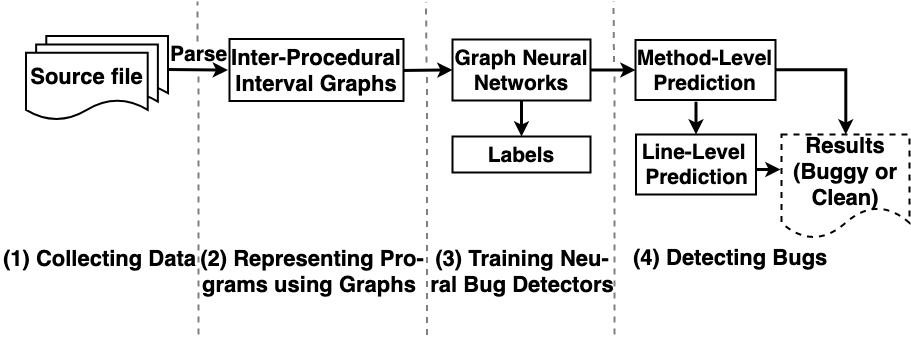}
  \caption{\tool's Workflow.}
  \label{fig:framework}
\end{figure}

Figure \ref{fig:framework} depicts the overview of \tool. We split \tool's workflow into four parts and describe each one below.

\begin{enumerate}[label*=(\arabic*)]
    \item \textbf{Collecting Data:} We extract programs from the codebases of real-world projects. We obtain bug labels by cross-referencing bug reports and the commit history of the same project. 
    
    \item \textbf{Representing Programs Using Graphs:} We construct the control flow graph for each method. To facilitate \tool to pinpoint a bug in a method, 
    we break each node denoting a basic block into multiple nodes each of which represents a non-control statement. We also stitch the control flow graphs of caller and callee to aid inter-procedural analysis.

    
    \item \textbf{Training Neural Bug Detectors:} We train neural bug detectors using graph neural networks. In particular, we propose a novel interval-based propagation mechanism that significantly enhances the capacity of GNN in learning programming patterns. 

    \item \textbf{Bug Detection:} Finally, we detect bugs in unseen code using the model we have trained.

\end{enumerate}

\subsection{Data Generation}
\label{subsec:data}
We use bug databases proposed by the prior works~\cite{just2014defects4j,saha2018bugs,tomassi2019bugswarm,ye2014learning} to generate our dataset. All databases provide a rich set of meta-information on each bug they contain (\eg detailed bug descriptions, links to actual commits of bug fixes, projects bug reports, \etc). We rely on such information to extract buggy methods and identify bug types. For example, given a bug description "\textit{Fixing Bugzilla\#2206...}", we search in the bug reports using the id \textit{\#2206} to retrieve the details of the bug, against which we match pre-defined keywords to determine the type of the bug. Next, we refer to the code commit to find out bug locations, specifically, we consider the lines that are modified by this commit to be buggy and the rest to be clean. We acknowledge the location of a bug and its fix may not be precisely the same thing. However, given the quality and maturity of the codebases we work with, our assumptions are reasonable. After extracting a buggy method, we choose from the same project several syntactically closest correct methods \wrt the tree edit distance between their abstract syntax trees.
 

\subsection{Representing Programs with Graphs}

We represent each method using its control flow graph. Besides, we split a graph node representing a basic block into a number of nodes, each of which represents a single statement. This variation enables the localization of bugs at the level of lines. To further enhance the expressiveness of our graph representation, we take the following measures: (1) we incorporate a new type of edges to represent the data dependency among variables. Specifically, we connect nodes with data dependency edges if the statements they represent exhibit a dependency relationship. Note that data dependency edges do not contribute to partitioning graphs into intervals. They only facilitate the local propagation within each interval. (2) we also add variable types into the mix. Given a variable $var$, we consider not only its actual type $\tau(var)$ but also all the supertypes of $\tau(var)$, denoted by $\tau^*(var) = \bigcup_{i=1}^{i=n} \tau^i(var)$, where $\tau^i(var) = \{\tau | \bigcup_{\omega \in \tau^{i-1}(var)} \omega$ implements $\tau\}$ with a base case $\tau^1(v) = \{\tau| \tau(var)$ implements type $\tau\}$. $\tau^n(var)$ signals the most general type (\eg \texttt{Object} in Java).

Due to the static nature of each non-control statement, we resort to recurrent neural networks---suitable for modeling sequentially data---to learn the initialization for each node in the graph. In particular, we treat a statement as a sequence of tokens. After each token is embedded into a numerical vector (similar to word embeddings~\cite{mikolov2013efficient}), we feed the entire token sequence into the network and extract its final hidden state as the initial node embedding. 

We stitch the graphs of all callers (\ie methods invoking the target method) and callees (\ie methods invoked by the target method) to that of each target method to support inter-procedural analysis. In particular, we connect the node of method invocation inside of the caller to that of the entry point inside of the callee. Node labels obtained for each graph will be kept on the merged graph for neural bug detectors to predict. Worth mentioning our graph representation is independent of the type of the bug detector we train. In other words, a graph only presents the semantic representation of a program, it's the responsibility of a bug detector to learn the patterns \wrt the type of bugs it's interested in.

\subsection{Training and Testing Neural Bug Detectors}
Given the method representations we derived, we train a classifier using GNN with IBPM. In particular, different types of bug detectors will be trained and applied separately to catch bugs in the unseen code. Internally, each neural bug detector is trained in two steps. First, we train a feed-forward neural network~\cite{svozil1997introduction} on top of method embeddings computed from GNN (\ie aggregation of the node embeddings) to predict if a method is buggy or not. If it is, we will use the same network to predict the buggy lines. Our rationale is graph as a whole provides stronger signals for models to determine the correctness of a method. Correct methods will be refrained from further predictions, resulted in fewer false warnings to be emitted. On the other hand, if a method is indeed predicted to be buggy, our model will then pick the top $N$ statements according to their probabilities of being buggy. The benefit of this prediction mode is to provide flexibility for developers to tune their analysis toward either producing less false warnings or identifying more bugs. 

\section{Evaluation}\label{sec:evaluation}
This section presents an extensive evaluation of the three neural bug detectors we instantiate from \tool. First, we evaluate the performance of each neural bug detector. Next, we highlight the improved model capacity IBPM attributes to. Finally, we show how our neural bug detector fares against several static analysis tools in catching null pointer dereference bugs, followed by a summary of applying \tool to catch bugs in the wild.

\subsection{Implementation}\label{subsec:impl}
Due to the popularity and availability of Java datasets, we target semantic bugs in Java programs. For a different language, \tool only requires a front-end analyzer that extracts the control flow graphs for programs in the new language, the rest of the workflow is language-agnostic. We construct the control flow and interval graphs using Spoon~\cite{pawlak:hal-01169705}, an open-source library for analyzing Java source code.  To efficiently choose correct methods to pair with each buggy method, we run DECKARD~\cite{4222572}, a clone detection tool adopting an approximation of the shortest tree edit distance algorithm to identify similar code snippets. 

We realize our IBPM-based GNN based on GGNN's implementation~\cite{li2015gated}. All neural bug detectors are implemented in Tensorflow. All RNNs built in the model have 1 recurrent layer with 100 hidden units. Each token in the vocabulary is embedded into a 100-dimensional vector. We experimented with other model configurations (\eg more recurrent layers, hidden units, \etc) and did not obtain better results. We use random initialization for weight initialization. We train the neural bug detectors using the Adam optimizer~\cite{kingma2014adam}. All experiments are performed on a 3.7GHz i7-8700K machine with 32GB RAM and NVIDIA GTX 1080 GPU.

\begin{table*}[h]
\caption{Details about our dataset. NPE, AIE, and CCE stand for null pointer dereference, array index out of bound and class cast exceptions respectively. Numbers in parenthesis records the number of synthetic bugs we created for AIE and CCE.}
\adjustbox{max width=.8\paperwidth}{
\centering
\small
\begin{tabular}{c|c|c|c|ccc}
\hline
 \multirow{2}{*}{Dataset} & \multirow{2}{*}{Projects} & \multirow{2}{*}{Description}& \multirow{2}{*}{Size (KLoC)}&  \multicolumn{3}{c}{Number of Buggy Methods} \\
\cline{5-7}
& & & &\tabincell{c}{NPE} & \tabincell{c}{AIE} & \tabincell{c}{\tabincell{c}{CCE}} \\\hline
\multirow{13}{*}{Test} &Lang & Java lang library &50 & 7 & 5 & 2  \\\cline{2-7}
&Closure & A JavaScript checker and optimizer. &260 & 18 & 0 & 0  \\\cline{2-7}
&Chart & Java chart library &149 & 17 & 7 & 0  \\\cline{2-7}
&Mockito & Mocking framework for unit tests &45 & 14 & 5 & 23  \\\cline{2-7}
&Math & Mathematics and statistics components &165 & 16 & 17 & 3  \\\cline{2-7}
&Accumulo & Key/value store &194 & 6 & 0 & 0  \\\cline{2-7}
&Camel & Enterprise integration framework &560 & 17 & 1 & 1  \\\cline{2-7}
&Flink & System for data analytics in clusters &258 & 13 & 0 & 0  \\\cline{2-7}
&Jackrabbit-oak & hierarchical content repository &337 & 23 & 1 & 1  \\\cline{2-7}
&Log4j2 & Logging library for Java &70 & 29 & 1 & 2  \\\cline{2-7}
&Maven & Project management and comprehension tool & 62 & 6 & 0 & 0  \\\cline{2-7}
&Wicket & Web application framework &206 & 7 & 1 & 8  \\\cline{2-7}
&Birt & Data visualizations platform &1,093 & 678 & 129 & 154  \\\hline
\hline
\multicolumn{4}{r|}{Number of Buggy Methods in Total for Test}& 793 & 167 & 192  \\ 
\multicolumn{4}{r|}{Number of Methods in Total for Test}& 3,000 & 600 & 800  \\\hline 
\hline
\multirow{6}{*}{Training}
&JDT UI & User interface for the Java IDE & 508 & 897 & 63 & 156 \\\cline{2-7}
&SWT & Eclipse Platform project repository & 460 & 276 & 136 & 26  \\\cline{2-7}
&Platform UI & User interface and help components of Eclipse & 595 & 920 & 72 & 194  \\\cline{2-7}
&AspectJ & An aspect-oriented programming extension &289 & 151 & 14 &  19 \\\cline{2-7}
&Tomcat & Web server and servlet container &222 & 111 & 18 & 28  \\\cline{2-7}
&from BugSwarm & 61 projects on GitHub & 5,191 & 156 & 20 & 3 \\\cline{2-7}
&from BugSwarm & 250 projects on GitHub & 19,874 & (0) & (2,356) & (1,998) \\\hline
\hline
\multicolumn{4}{r|}{Number of Buggy Methods in Total for Training}& 2,511 & 2,679 & 2,424  \\
\multicolumn{4}{r|}{Number of Methods in Total for Training}& 10,000 & 10,000 & 10,000  \\\hline
\end{tabular}
}
\label{tab:stas}
\end{table*}

\subsection{Experiment Setup}
We assemble a dataset according to the procedure described in Section~\ref{subsec:data}.
Table~\ref{tab:stas} shows the details: the projects from which we extract the code snippet, a brief description of their functionality, size of their codebase, and the number of buggy methods we extract from each project. We maintain approximately 3:1 ratio between the correct and buggy methods across each bug type within each project.
To enrich the set of AIE (Array Index Out of Bound Exception) and CCE (Class Cast Exception) bugs, we create synthetic buggy examples based on a collection of highly-rated and actively maintained projects from BugSwarm~\cite{tomassi2019bugswarm}. In particular, we first randomly pick methods that contain either array access or class cast operations, and then strip away the bounds or type checks. The number in parenthesis depicts how many buggy examples we synthesized for AIE and CCE. Note that all synthetic bugs are used for the training purpose only, and all bugs in the test set are real. We also make certain clean methods that are paired with synthetic buggy examples are chosen from other methods rather than their original version (\ie with the bounds or type checks). We evaluate \tool on a standard classification task, that is, predicting each line in a method in the test set to be buggy or not. Since we deal with a largely unbalanced dataset, we forgo the accuracy metric and opt for Precision, Recall, and F1 Score, metrics are commonly used in defect prediction literature~\cite{wang2016automatically,pradel2018deepbugs}.

As a baseline, we select standard GGNN~\cite{li2015gated,allamanis2017learning}, a predominate deep learning model in the domain of programming. In particular, we re-implemented the program encoding scheme proposed by~\citet{allamanis2017learning}, a graph representation built of top of the AST with additional edges denoting variable types, data, control flow dependencies, \etc

\begin{table*}
\begin{minipage}{.33\textwidth}
\captionsetup{skip=1pt}
\caption{Precision in predicting NPE.}
\centering
\adjustbox{max width=1\textwidth}{
\begin{tabular}{c|c|c|c}
\hline
 \tabincell{c}{Methods} & Top-1 & Top-3 & Top-5  \\\hline
Baseline (0) & 0.209 & 0.147 & 0.117 \\\hline
Baseline (1) & \textit{0.285} & 0.166 & 0.131 \\\hline
Baseline (2) & 0.218 & 0.138 & 0.101 \\\hline
\tool (0) & 0.326 & 0.174 & 0.138 \\\hline
\tool (1) & \textbf{0.351} & 0.167 & 0.130 \\\hline
\tool (2) & 0.305 & 0.164 & 0.136 \\\hline
\hline
Gain & \multicolumn{3}{|l}{\hspace{1pt}\textbf{0.351} - \textit{0.285} = 0.066} \\\hline
\end{tabular}}
\label{tab:preNPE}
\end{minipage}
\begin{minipage}{.33\textwidth}
\captionsetup{skip=1pt}
\caption{Recall in predicting NPE.}
\centering
\adjustbox{max width=1\textwidth}{
\begin{tabular}{c|c|c|c}
\hline
 \tabincell{c}{Methods} & Top-1 & Top-3 & Top-5  \\\hline
Baseline (0) & 0.243 & 0.373 & \textit{0.450} \\\hline
Baseline (1) & 0.171 & 0.252 & 0.302 \\\hline
Baseline (2) & 0.210 & 0.337 & 0.375 \\\hline
\tool (0) & 0.282 & 0.374 & 0.447 \\\hline
\tool (1) & 0.329 & 0.431 & \textbf{0.507} \\\hline
\tool (2) & 0.304 & 0.402 & 0.500 \\\hline
\hline
Gain & \multicolumn{3}{|l}{\hspace{1pt}\textbf{0.507} - \textit{0.450} = 0.057} \\\hline
\end{tabular}}
\label{tab:reNPE}
\end{minipage}
\begin{minipage}{.33\textwidth}
\captionsetup{skip=1pt}
\caption{F1 score in predicting NPE.}
\centering
\adjustbox{max width=1\textwidth}{
\begin{tabular}{c|c|c|c}
\hline
 \tabincell{c}{Methods} & Top-1 & Top-3 & Top-5  \\\hline
Baseline (0) & \textit{0.224} & 0.211 & 0.186 \\\hline
Baseline (1) & 0.214 & 0.201 & 0.183 \\\hline
Baseline (2) & 0.214 & 0.196 & 0.159 \\\hline
\tool (0) & 0.303 & 0.238 & 0.211 \\\hline
\tool (1) & \textbf{0.339} & 0.241 & 0.207 \\\hline
\tool (2) & 0.304 & 0.233 & 0.214 \\\hline
\hline
Gain & \multicolumn{3}{|l}{\hspace{1pt}\textbf{0.339} - \textit{0.224} = 0.115} \\\hline
\end{tabular}}
\label{tab:f1NPE}
\end{minipage}

\par\bigskip

\begin{minipage}{.33\textwidth}
\captionsetup{skip=1pt}
\caption{Precision in predicting AIE.}
\centering
\adjustbox{max width=1\textwidth}{
\begin{tabular}{c|c|c|c}
\hline
 \tabincell{c}{Methods} & Top-1 & Top-3 & Top-5  \\\hline
Baseline (0) & \textit{0.286} & 0.124 & 0.106 \\\hline
Baseline (1) & 0.218 & 0.112 & 0.095 \\\hline
Baseline (2) & 0.215 & 0.134 & 0.088 \\\hline
\tool (0) & 0.269 & 0.159 & 0.113 \\\hline
\tool (1) & 0.308 & 0.134 & 0.097 \\\hline
\tool (2) & \textbf{0.362} & 0.131 & 0.110 \\\hline
\hline
Gain & \multicolumn{3}{|l}{\hspace{1pt}\textbf{0.362} - \textit{0.286} = 0.076} \\\hline
\end{tabular}}
\label{tab:preAI}
\end{minipage}\hfill
\begin{minipage}{.33\textwidth}
\captionsetup{skip=1pt}
\caption{Recall in predicting AIE}
\centering
\adjustbox{max width=1\textwidth}{
\begin{tabular}{c|c|c|c}
\hline
 \tabincell{c}{Methods} & Top-1 & Top-3 & Top-5  \\\hline
Baseline (0) & 0.231 & 0.266 & \textit{0.355} \\\hline
Baseline (1) & 0.165 & 0.225 & 0.300 \\\hline
Baseline (2) & 0.180 & 0.288 & 0.306 \\\hline
\tool (0) & 0.284 & 0.425 & 0.470 \\\hline
\tool (1) & 0.291 & 0.358 & 0.403 \\\hline
\tool (2) & 0.333 & 0.378 & \textbf{0.489} \\\hline
\hline
Gain & \multicolumn{3}{|l}{\hspace{1pt}\textbf{0.489} - \textit{0.355} = 0.134} \\\hline
\end{tabular}}
\label{tab:reAI}
\end{minipage}\hfill
\begin{minipage}{.33\textwidth}
\captionsetup{skip=1pt}
\caption{F1 score in predicting AIE.}
\centering
\adjustbox{max width=1\textwidth}{
\begin{tabular}{c|c|c|c}
\hline
 \tabincell{c}{Methods} & Top-1 & Top-3 & Top-5  \\\hline
Baseline (0) & \textit{0.256} & 0.169 & 0.163 \\\hline
Baseline (1) & 0.188 & 0.149 & 0.145 \\\hline
Baseline (2) & 0.196 & 0.183 & 0.136 \\\hline
\tool (0) & 0.276 & 0.232 & 0.182 \\\hline
\tool (1) & 0.299 & 0.195 & 0.157 \\\hline
\tool (2) & \textbf{0.347} & 0.195 & 0.179 \\\hline
\hline
Gain & \multicolumn{3}{|l}{\hspace{1pt}\textbf{0.347} - \textit{0.256} = 0.091} \\\hline
\end{tabular}}
\label{tab:f1AI}
\end{minipage}\hfill

\par\bigskip

\begin{minipage}{.33\textwidth}
\captionsetup{skip=1pt}
\caption{Precision in predicting CCE.}
\centering
\adjustbox{max width=1\textwidth}{
\begin{tabular}{c|c|c|c}
\hline
 \tabincell{c}{Methods} & Top-1 & Top-3 & Top-5  \\\hline
Baseline (0) & 0.345 & 0.165 & 0.169 \\\hline
Baseline (1) & 0.328 & 0.180 & 0.135 \\\hline
Baseline (2) & \textit{0.352} & 0.198 & 0.173 \\\hline
\tool (0) & 0.362 & 0.220 & 0.170 \\\hline
\tool (1) & \textbf{0.438} & 0.244 & 0.199 \\\hline
\tool (2) & \textbf{0.438} & 0.252 & 0.193 \\\hline
\hline
Gain & \multicolumn{3}{|l}{\hspace{1pt}\textbf{0.438} - \textit{0.387} = 0.051} \\\hline
\end{tabular}}
\label{tab:preCC}
\end{minipage}\hfill
\begin{minipage}{.33\textwidth}
\captionsetup{skip=1pt}
\caption{Recall in predicting CCE.}
\centering
\adjustbox{max width=1\textwidth}{
\begin{tabular}{c|c|c|c}
\hline
 \tabincell{c}{Methods} & Top-1 & Top-3 & Top-5  \\\hline
Baseline (0) & 0.270 & 0.319 & 0.467 \\\hline
Baseline (1) & 0.319 & 0.369 & 0.393 \\\hline
Baseline (2) & 0.291 & 0.371 & \textit{0.503} \\\hline
\tool (0) & 0.286 & 0.443 & 0.495 \\\hline
\tool (1) & 0.292 & 0.413 & 0.486 \\\hline
\tool (2) & 0.384 & 0.561 & \textbf{0.620} \\\hline
\hline
Gain & \multicolumn{3}{|l}{\hspace{1pt}\textbf{0.620} - \textit{0.503} = 0.117} \\\hline
\end{tabular}}
\label{tab:reCC}
\end{minipage}\hfill
\begin{minipage}{.33\textwidth}
\captionsetup{skip=1pt}
\caption{F1 score in predicting CCE.}
\centering
\adjustbox{max width=1\textwidth}{
\begin{tabular}{c|c|c|c}
\hline
 \tabincell{c}{Methods} & Top-1 & Top-3 & Top-5  \\\hline
Baseline (0) & 0.303 & 0.217 & 0.249 \\\hline
Baseline (1) & \textit{0.323} & 0.242 & 0.201 \\\hline
Baseline (2) & 0.319 & 0.258 & 0.257 \\\hline
\tool (0) & 0.320 & 0.294 & 0.254 \\\hline
\tool (1) & 0.350 & 0.307 & 0.282 \\\hline
\tool (2) & \textbf{0.409} & 0.348 & 0.295 \\\hline
\hline
Gain & \multicolumn{3}{|l}{\hspace{1pt}\textbf{0.409} - \textit{0.323} = 0.086} \\\hline
\end{tabular}}
\label{tab:f1CC}
\end{minipage}\hfill
\end{table*}

\ignore{
\begin{table*}
\begin{minipage}{.33\textwidth}
\captionsetup{skip=1pt}
\caption{Precision in predicting all bugs.}
\centering
\adjustbox{max width=1\textwidth}{
\begin{tabular}{c|c|c|c}
\hline
 \tabincell{c}{\#Interval} & Top-1 & Top-3 & Top-5  \\\hline
1-2 & 51.1/40.2 & 29.7/23.2 & 24.0/18.1 \\\hline
3-4 & 10.0/23.1 & 8.9/12.3 & 8.7/9.2 \\\hline
5-6 & 4.3/0.0 & 2.9/6.0 & 1.7/4.8 \\\hline
>6 & 8.3/0.0 & 2.8/0.0 & 3.3/0.7 \\\hline
\end{tabular}}
\label{tab:cfgpre}
\end{minipage}
\begin{minipage}{.33\textwidth}
\captionsetup{skip=1pt}
\caption{Recall in predicting all bugs.}
\centering
\adjustbox{max width=1\textwidth}{
\begin{tabular}{c|c|c|c}
\hline
 \tabincell{c}{\#Interval} & Top-1 & Top-3 & Top-5  \\\hline
1-2 & 30.2/37.1 & 40.4/49.9 & 46.7/55.9 \\\hline
3-4 & 9.7/21.4 & 25.8/34.3 & 41.9/42.9 \\\hline
5-6 & 4.5/0.0 & 9.1/26.0 & 9.1/34.6 \\\hline
>6 & 8.3/0.0 & 8.3/0.0 & 16.7/3.3 \\\hline
\end{tabular}}
\label{tab:cfgre}
\end{minipage}
\begin{minipage}{.33\textwidth}
\captionsetup{skip=1pt}
\caption{F1 score in predicting all bugs.}
\centering
\adjustbox{max width=1\textwidth}{
\begin{tabular}{c|c|c|c}
\hline
 \tabincell{c}{\#Interval} & Top-1 & Top-3 & Top-5  \\\hline
1-2 & 37.9/38.6 & 34.2/31.7 & 31.7/27.4 \\\hline
3-4 & 9.8/22.2 & 13.2/18.1 & 14.4/15.2 \\\hline
5-6 & 4.4/NA & 4.4/9.7 & 2.9/8.4 \\\hline
>6 & 8.3/NA & 4.2/NA & 5.6/1.1 \\\hline
\end{tabular}}
\label{tab:cfgf1}
\end{minipage}

\par\bigskip

\end{table*}
}

\subsection{Performance Measurement}
We build two sets of neural bug detectors based on \tool and the baseline. Each set includes three models, each of which deals with null pointer dereference, array index out of bound and class casting bugs respectively. 

\paragraph{Evaluating models using proposal metrics} Table~\ref{tab:preNPE}--\ref{tab:f1CC} depicts the precision, recall, and F1 for all neural bug detectors. We set out to examine the impact of the depth of inlining on the performance of each bug detector. In particular, each row in a table shows the result of a neural bug detector that is trained either using the baseline model or \tool. The number in parenthesis is the depth of inlining. Specifically, 0 means no inlining, 1 inlines all methods that are invoking or invoked by the target methods, and 2 further inlines all callers and callees of each method that is inlined by 1. Exceeding the depth of 2 has costly consequences. First, most data points will have to be filtered out for overflowing the GPU memory. Moreover, a significant portion of the remaining ones would also have to be placed in a batch on its own, resulted in a dramatic increase in training time. Methods in the test set are inlined as deep as those in the training set to maintain a consistent distribution across the entire dataset. We also make certain each buggy method is only inlined once to prevent a potential date duplicate issue.
Columns in each table correspond to three prediction modes in which \tool picks 1, 3 or 5 statements to be buggy after a method is predicted to be buggy.

Overall, \tool consistently outperforms the baseline across all bug types using all proposed metrics. In particular, \tool beats the baseline by around 10\% across all bugs types in F1 score. 
Regarding the impact of the depth of inlining, we find that increasing the depth mostly but not always leads to improved performance of either \tool or the baseline. The reason is, on one hand, models will benefit from a more complete program representation thanks to the inlining. On the other hand, exceedingly large graphs hinders the generalization of GNNs, resulted in degraded performance of the neural bug detectors.

\paragraph{Investigating models scalability}
We divide the entire test set into ten subsets, each of which consists of graphs of similar size. We then investigate how models perform on each subset starting from the smallest to the largest graphs. We fix the depth of inlining to be 2, which provides the largest variation among all graphs in terms of the size. Figure~\ref{fig:scaf} depicts the F1 score for \tool and the baseline under top-3 prediction mode (top-1 and top-5 yield similar results). Initially, \tool achieves slightly better results on the subset of small graphs; later as the size of graphs increases, \tool is significantly more resilient with its prediction than the baseline. Overall, we conclude bug detectors built out of \tool are more scalable.

\begin{figure*}
\begin{minipage}{.33\textwidth}
\centering
\begin{tikzpicture}[scale=0.92]
   \begin{axis}[
    xmax=100,xmin=0,
    ymin=0,ymax=0.7,
    xlabel=\emph{Graph Size},ylabel=\emph{Top-3 F1 Score},
    xtick={0,20,40,...,100},
    ytick={0,0.1,0.20,...,0.7},
    xticklabel={\pgfmathparse{\tick}\pgfmathprintnumber{\pgfmathresult}\%},
    ]
     \addplot+ [mark=pentagon*,mark size=2.5pt,every mark/.append style={}] coordinates{(5, 0.481) (15, 0.385) (25, 0.288) (35, 0.256) (45, 0.256) (55, 0.163) (65, 0.156) (75, 0.127) (85, 0.070) (95, 0.027)};
     \addplot coordinates{(5, 0.388) (15, 0.294) (25, 0.350) (35, 0.309) (45, 0.303) (55, 0.266) (65, 0.275) (75, 0.244) (85, 0.149) (95, 0.073)};   
    \legend{\emph{Baseline},\emph{\tool}}
    \end{axis}
    \end{tikzpicture}
     ~%
    \subcaption{}
\label{fig:scaf}
\end{minipage}
\begin{minipage}{.33\textwidth}
\centering
\begin{tikzpicture}[scale=0.92]
   \begin{axis}[
    xmax=120,xmin=0,
       ymin=0,ymax=0.4,
    xlabel=\emph{Training Time (Minutes)},ylabel=\emph{Average Test F1 for NPE},
    xtick={0,30,60,90,120},
    ytick={0,0.1,0.20,...,0.4},
    xticklabel={\pgfmathparse{\tick}\pgfmathprintnumber{\pgfmathresult}},
   ]
   
     \addplot+ [mark=pentagon*,mark size=2.5pt,every mark/.append style={}] coordinates{(0, 0.1)  (15, 0.123)  (30, 0.129) (45, 0.137)  (60, 0.169)  (75, 0.167)  (90, 0.174)  (105, 0.19)  (120, 0.191)  (135, 0.239)  (150, 0.24)};
     \addplot coordinates{(0, 0.117)  (15, 0.185)  (30, 0.204)  (45, 0.232) (60, 0.252)  (75, 0.255)  (90, 0.255)  (105, 0.256)  (120, 0.255) (135, 0.412) (150, 0.41)};
     \legend{\emph{Baseline}, \emph{\tool}}
    \end{axis}
    \end{tikzpicture}%
    ~%
\subcaption{}
\label{fig:ttf}
\end{minipage}
\begin{minipage}{.33\textwidth}
\vspace{4pt}
\centering
\begin{tikzpicture}[scale=0.92]
  \begin{axis}[
    ymin=0,ymax=0.7,
    xlabel=\emph{Number of First Order Intervals}, ylabel=\emph{Top-3 F1 Score},
    xtick=data,
    xticklabels={1,2,...,6,7+},
    ybar,
    bar width=6pt,
    ytick pos=left,
    xtick style={draw=none},
    xticklabel style={yshift=.75ex},
    legend image code/.code={
        \draw [#1] (0cm,-0.1cm) rectangle (0.2cm,0.25cm); },
  ]

  \addplot [black!30!white,fill=black!30!white] coordinates {
  (1, 0.257) (2, 0.215) (3, 0.183) (4, 0.173) (5, 0.116) (6, 0.115) (7, 0.074)
  };
  \addplot [black!100!white,fill=black!100!white] coordinates {
  (1, 0.352) (2, 0.288) (3, 0.259) (4, 0.291) (5, 0.102) (6, 0.159) (7, 0.107)
  };
  \legend{\emph{w/o IBPM},\emph{\tool}}
  \end{axis}
\end{tikzpicture}%
    ~%
\subcaption{}
\label{fig:intf}
\end{minipage}
\caption{Experiment results.}
\end{figure*}
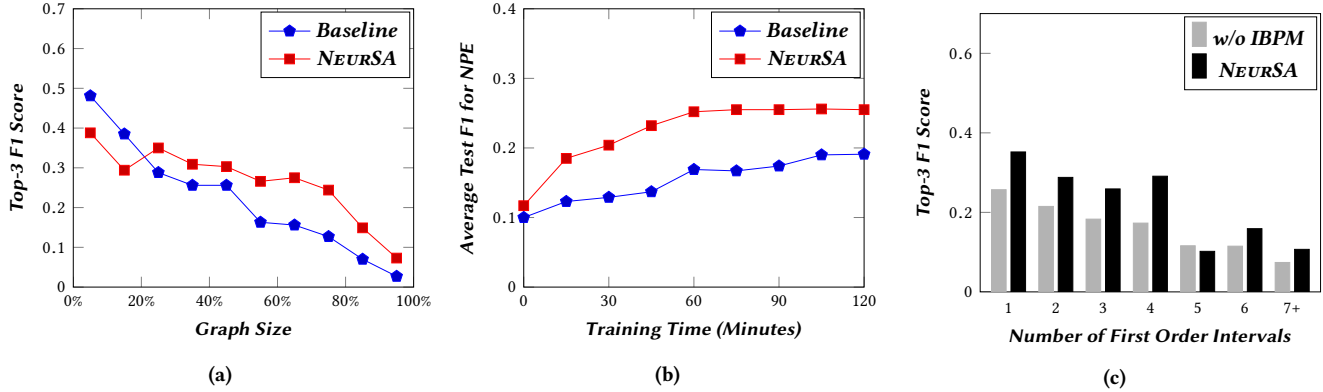

\paragraph{Measuring training efficiency}
Since all models make instantaneous predictions (\eg around 800 examples per second), we only measure the training speed of \tool and the baseline. Again, we fix the depth of inlining to be 2, which presents the biggest challenges in training efficiency. Figure~\ref{fig:ttf} shows the \textit{test} F1 score (average across top-1, 3 and 5 prediction modes) over training time for \tool and the baseline on null pointer dereference bugs (Other two bug types are much easier to train). \tool's bug detector achieves results that are as 50\% as good as its final results under half an hour, and results that are as 95\% as good under 1 hour, while both being substantially higher than the best results of the baseline models. Our model achieves its best results around the 1 hour mark.

\subsection{Contribution of IBPM}
We aim to precisely quantify IBPM's contribution to the learning capacity of the GNN. In particular, we re-implement \tool by replacing its IBPM with the standard propagation mechanism. In other words, we perform a head-to-head comparison between the two propagation mechanisms while the rest of the model architecture is held the same. Figure~\ref{fig:intf} shows how the test F1 scores of both configurations vary with the number of first order intervals in top-3 prediction mode (top-1 and top-5 yield similar results). For the same reason, the depth of inlining is fixed to be 2. Given the same number of the first order intervals, IBPM performs better than the standard propagation mechanism across the board. Apart from the apparent scalability advantage IBPM displays over the standard propagation mechanism, IBPM achieves better results even when the number of intervals is small, thanks to its exploitation of the program-specific graph characteristics.

\subsection{Futher Evaluation of \tool}
In general, static analyzers can be classified into two categories: soundy and unsound. The former refers to tools that are sound by design but make unsound choices for utility concerns, whereas the latter is often built out of pattern matching techniques that do not concern themselves with soundness. Intending to cover both flavors, we pick soundy tools: 
Facebook Infer~\cite{calcagno2015moving,berdine2005smallfoot}, Pinpoint~\cite{shi2018pinpoint} and BugPicker, and another unsound tool: SpotBugs (the successor of FindBugs~\cite{Hovemeyer2004}).

We emphasize that our work makes no claims of advancing the theory of sound static analysis, therefore this experiment solely focuses on studying the utility of each tool in practice. 
We examine the performance of each checker in catching null pointer dereference bugs, which happen to be the only type of bugs all tools support.





To keep the engineering load manageable, we randomly pick 75 bugs out of seven projects from defect4j and bugs.jar, two most studied datasets in bug detection literature. Later we discover Infer does not have access to the build tools that are required to compile some of the projects in defect4j. Therefore, we show Infer's performance only on the projects it can compile. Those projects contain 58 bugs in total, namely Bug Set I.

First, we manually verify the validity, and more importantly, the location for all 75 bugs in case any bug may trigger an exception at a different location than developers' fix.
Next, we feed the entire codebase of each project for each checker to scan. A report that points to a buggy location will be counted as a true positive, otherwise, it's a false positive. We acknowledge that a static analysis tool could have caught a real bug that is mistakenly counted as a false positive since it falls out of the 75 bugs we collected, however, the exact same issue applies to \tool, which we regrettably ignore due to resource constraint. We use \tool with depth 2 of inlining trained in the previous experiment. It is configured to run the top-1 prediction mode. We choose the most precise analysis for Pinpoint to perform. We also manually fine-tune BugPicker (\eg increase the length of the call chain or reset the timeout), but obtain the same results. We do not manage to run any other configuration for Infer and SpotBugs rather than their default set up. 

Table~\ref{tab:sats} depicts the results. Overall, we find (1) \tool achieves better precision and recall than any other evaluated tool; more importantly, (2) bugs that are found by competing tools are all caught by \tool. Besides, \tool detects another 11 bugs that are missed by the other tools.

\begin{table}[!t]
\centering
\normalsize
\caption{Results of the comparison. TP, FP, FN stand for true positives, false positives and false negatives respectively.}
\adjustbox{max width=\columnwidth}{\begin{tabular}{c|ccc|ccc|cc}
\hline
 \multirow{2}{*}{\tabincell{c}{Tools}} & \multicolumn{3}{c|}{Bug Set I} & \multicolumn{3}{c|}{Remaining Set} & \multicolumn{2}{c}{Total}\\
\cline{2-9}
         & TP & FP & FN & TP & FP & FN  & Precision & Recall   \\\hline
FB Infer & 4 & 1004 & 54 & NA & NA & NA & .004 & .069    \\\hline
Pinpoint & 4 & 1524 & 54 & 3 & 1104 & 14 & .003 & .093   \\\hline
SpotBugs & 0 & 107 & 58 & 1 & 26 & 16 & .007 & .013   \\\hline
BugPicker & 0 & 79 & 58 & 0 & 0 & 17 & 0 & 0   \\\hline
\hline
\textbf{\tool} & \textbf{16} & \textbf{1655} & \textbf{42} & \textbf{5} & \textbf{521} & \textbf{12} & \textbf{.010} & \textbf{.280}  \\\hline 

\end{tabular}}
\label{tab:sats}
\end{table}

We have also applied \tool to catch bugs in the wild, and obtained encouraging results. Configured with top-1 prediction mode and inline depth of 2 (we also tried 3 and 4 for smaller methods), \tool found \rBugs bugs (40 NPE and 10 AIE) in 17 projects that are among the 
most starred on GitHub. All bugs are manually verified and reported to developers, out of which \fBugs are fixed and another \cBugs are confirmed (fixes pending). Our findings show that not only does \tool significantly improve the prior works, but more importantly, \tool is a industry-strength static tool in checking real-world Java projects.

\subsection{Discussion on Limitation of \tool}
In order to gain a deeper understanding of \tool's limitation, we mainly look into the bugs \tool missed in the prior experiment. We find that when bug patterns are not local, that is, the dereference is quite distant from where a null pointer is assigned, \tool tends to struggle, in particular, it even misclassifies the method to be clean most of the time. The reason is, as one can imagine, \tool's signal gets diluted when a scattered bug pattern is overwhelmed by irrelevant statements in a program. We plan to tackle this issue with both program analysis and machine learning techniques. For data pre-processing, program slicing~\cite{Weiser1981} seems to be a promising idea. Given the location of a potential null pointer dereference (due to the null assignment or non-initialization), we can filter out statements that fall out of the slice to localize a bug pattern. On the other hand, a larger, more diverse dataset that contains a variety of bug patterns should also help, or for data augmentation, we can synthesize new programs by randomly adding dummy statements in between the assignment and dereference of null pointer for buggy methods, especially when the high-quality, large-scale public datasets are unavailable.

Another issue we discovered is \tool is less accurate when handling large programs (\ie approximately 1,000 nodes in the graph) despite the contribution IBPM makes to the underlying GNN. We find several bugs \tool missed where the bug patterns are in fact local. We believe this is due to the fundamental limitation of GNN, which requires deeper insight and understanding of the models themselves, without which, IBPM's results will be hard to improve upon in a qualitative manner.

We also briefly examine the false positives \tool emitted. For a few dozen we inspected, we find that the same issue causing \tool to miss bugs are also major sources of \tool's imprecision. In addition, \tool can be over-sensitive about the presence of null pointers, albeit we can not confirm they are never dereferenced.

\ignore{

\begin{figure*}
\centering
\begin{tikzpicture}[scale=0.92]
  \begin{axis}[
    ymin=0,ymax=0.7,
    xlabel=\emph{Number of First Order Intervals}, ylabel=\emph{Top-1 F1 Score},
    xtick=data,
    xticklabels={1,2,...,6,7+},
    ybar,
    bar width=6pt,
    legend image code/.code={
        \draw [#1] (0cm,-0.1cm) rectangle (0.2cm,0.25cm); },
  ]

  \addplot [black!30!white,fill=black!30!white] coordinates {
  (1, 0.297) (2, 0.220) (3, 0.202) (4, 0.213) (5, 0.058) (6, 0.095) (7, 0.044) 
  };
  \addplot [black!100!white,fill=black!100!white] coordinates {
  (1, 0.358) (2, 0.310) (3, 0.233) (4, 0.397) (5, 0.076) (6, 0.154) (7, 0.213)
  };
  \legend{\emph{w/o IBPM},\emph{\tool}}
  \end{axis}
\end{tikzpicture}%
    ~%
\begin{tikzpicture}[scale=0.92]
  \begin{axis}[
    ymin=0,ymax=0.7,
    xlabel=\emph{Number of First Order Intervals}, ylabel=\emph{Top-3 F1 Score},
    xtick=data,
    xticklabels={1,2,...,6,7+},
    ybar,
    bar width=6pt,
    legend image code/.code={
        \draw [#1] (0cm,-0.1cm) rectangle (0.2cm,0.25cm); },
  ]

  \addplot [black!30!white,fill=black!30!white] coordinates {
  (1, 0.257) (2, 0.215) (3, 0.183) (4, 0.173) (5, 0.116) (6, 0.115) (7, 0.074) 
  };
  \addplot [black!100!white,fill=black!100!white] coordinates {
  (1, 0.352) (2, 0.288) (3, 0.259) (4, 0.291) (5, 0.102) (6, 0.159) (7, 0.107)
  };
  \legend{\emph{w/o IBPM},\emph{\tool}}
  \end{axis}
\end{tikzpicture}%
    ~%
\begin{tikzpicture}[scale=0.92]
  \begin{axis}[
    ymin=0,ymax=0.7,
    xlabel=\emph{Number of First Order Intervals}, ylabel=\emph{Top-5 F1 Score},
    xtick={1,2,...,7},
    xticklabels={1,2,...,6,7+},
    ybar,
    bar width=6pt,
    legend image code/.code={
        \draw [#1] (0cm,-0.1cm) rectangle (0.2cm,0.25cm); },
  ]

  \addplot [black!30!white,fill=black!30!white] coordinates {
  (1, 0.215) (2, 0.182) (3, 0.136) (4, 0.138) (5, 0.090) (6, 0.114) (7, 0.048) 
  };
  \addplot [black!100!white,fill=black!100!white] coordinates {
  (1, 0.326) (2, 0.269) (3, 0.228) (4, 0.219) (5, 0.099) (6, 0.163) (7, 0.071)
  };
  \legend{\emph{w/o IBPM},\emph{\tool}}
  \end{axis}
\end{tikzpicture}
    \caption{Comparing IBPM against the standard propagation. \wyc{(11+ represent graphs with more than 11 intervals.)}}
    \label{fig:pm}
     
\end{figure*}

}

\section{Related Work}\label{sec:related}

\paragraph{\textbf{Machine Learning for Defect Prediction}}
Utilizing machine learning techniques for software defect prediction is another rapidly growing research field. So far the literature has been focusing on detecting simpler bugs that are syntactic in nature. 
\ignore{
Choi et al. train a memory network model \cite{} for predicting buffer overruns in programming language analysis\cite{Choi2017bufferoverrun}. 
Their work can only work on buffer overrun prediction since they generate dataset by only focusing on buffer related functions}
\citet{wang2016automatically} leverages deep belief network to learn program representations for defect prediction. Their model predicts an entire file to be buggy or not. \citet{pradel2018deepbugs} present Deepbugs, a framework to detect name-based bugs in binary operation. Another line of works targets variable misuse bugs~\cite{allamanis2017learning,vaswani2017attention}, in which developers use wrong variables, and models are required to predict the ones that should have been used. Clearly, \tool significantly differs from all of them as it locates complex semantic bugs.



\paragraph{\textbf{Static Bug Finding}}
In principle, static analysis models all executions of a program in order to provide the soundness guarantee. However, perfectly sound static analysis tools almost never exist; concessions that sacrifice the soundness guarantee have to be made to ensure the usability of static analyzers. Several years ago, a new term---soundiness~\cite{livshits2015defense}---is brought forward aiming to clarify the level of approximation adopted in the analysis. Below we survey several flagship static analyzers.

SLAM~\cite{Ball2002}, BLAST~\cite{Henzinger2002}, and SATABS~\cite{Clarke2004} adopt abstract
refinement for static bug checking. CBMC~\cite{Clarke2003,Clarke2004} performs bounded model checking. Saturn~\cite{Xie2005,Dillig2008} is another program analysis system, which uses the combination of summaries and constraints to achieve both high scalability and precision. Compass~\cite{Dillig2011} performs bottom-up, summary-based heap analysis for the verification of real C and C++ programs.

Compared with static analysis tools, we believe \tool yields several conceptual advantages. First, with \tool, creating a static bug finder can be simplified to solving a machine learning task. In comparison, the traditional way of designing static analyzers even the unsound ones requires a substantial amount of human expertise. Besides, unlike the static analysis tools, model-based bug detectors can benefit from increasingly large datasets and powerful training regime, as a result, they will evolve and get better over time.




\section{Conclusion}\label{sec:conclusion}

In this paper, we present \tool, a framework for creating neural bug detectors based on GNN. We also propose an interval-based propagation model to improve the capacity of GNN in learning program properties. Our evaluation shows that \tool is effective in catching semantic bugs; it outperforms several flagship static analyzers in catching null pointer dereference bugs; and even finds new bugs in many highly-rated and actively maintained projects on GitHub. For future work, we will apply the IBPM-Based GNN to other program analysis tasks as we believe our approach offers a general, and powerful framework for learning effective program analyzers.



\clearpage
\bibliography{NeurSA}




\end{document}